\documentclass[preprint,12pt]{elsarticle}

\usepackage{amssymb}
\usepackage{microtype}
\usepackage{times}  
\usepackage{helvet}  
\usepackage{courier}  
\usepackage{url}  
\usepackage{graphicx}  
\usepackage{amssymb}
\usepackage{amsmath}
\usepackage{algorithm}
\usepackage{algpseudocode}
\usepackage{amsmath}
\usepackage{graphics}
\usepackage{epsfig}
\usepackage{color}

\newtheorem{theorem}{Theorem}
\newtheorem{lemma}[theorem]{Lemma}

\newtheorem{corollary}[theorem]{Lemma}
\newtheorem{definition}[theorem]{Lemma}
\newdefinition{rmk}{Remark}
\newproof{proof}{Proof}

\newcommand{\be}{\begin{equation}}
\newcommand{\ee}{\end{equation}}
\def\bea{\begin{eqnarray}}
\def\eea{\end{eqnarray}}

\def\ba{\begin{array}}
\def\ea{\end{array}}

\begin{document}
\begin{frontmatter}

\title{A constant FPT approximation algorithm for hard-capacitated $k$-means}

\author[1]{Yicheng Xu}
\author[2]{Rolf H. M\"ohring}
\author[3]{Dachuan Xu}
\author[4]{Yong Zhang}
\author[5]{Yifei Zou}

\address[1]{Shenzhen Institutes of Advanced Technology, Chinese Academy of Sciences,  P.R. China}
\address[2]{Institut f\"ur Mathematik, Technische Universit\"at Berlin, Germany}
\address[3]{Department of Operations Research and Scientific Computing, Beijing University of Technology}
\address[4]{Shenzhen Institutes of Advanced Technology, Chinese Academy of Sciences,  P.R. China}
\address[5]{Department of Computer Science, The University of Hong Kong, P.R. China}

\begin{abstract}
Hard-capacitated $k$-means (HCKM) is one of the fundamental problems remaining open in combinatorial optimization and data mining areas. In this problem, one is required to partition a given $n$-point set into $k$ disjoint clusters with known capacity so as to minimize the sum of within-cluster variances. It is known to be at least APX-hard and for which most of the work is from a meta heuristic perspective. To the best our knowledge, no constant approximation algorithm or existence proof of such an algorithm is known. As our main contribution, we propose an FPT($k$) algorithm with performance guarantee of $69+\epsilon$ for any HCKM instances in this paper.
\end{abstract}
\begin{keyword}
Approximation algorithm, FPT approximation, capacitated clustering, $k$-means.
\end{keyword}
\end{frontmatter}

\section{Introduction}
The $k$-means (KM) is one of the most fundamental clustering tasks in combinatorial optimization and data mining. Given an $n$-point data set $X$ in $\mathbb{R}^d$ and an integer $k(\le n)$, the goal is to partition $X$ into $k$ disjoint subsets so as to minimize the total within-cluster variances or, equivalently speaking, to choose $k$ centers so as to minimize the total squared distances between each point and its closest center. This problem is NP-hard in general dimension even when $k=2$, proved by reduction from the densest cut problem \cite{adhp2009}. And the latest work shows that it is APX-hard \cite{acks2015} in general dimension
 and the inapproximability is at least 1.0013 \cite{lsw2016}. Despite its hardness, Lloyd \cite{l1982}
 proposes a local search heuristic for this problem that performs very well in practical and is still widely used today. Berkhin \cite{b2006}
 states in public that "it is by far the most popular clustering
algorithm used in scientific and industrial applications".

It turns out that in the literature, most of the work for $k$-means has been done from a practical perspective and is thus more time-concerned. Among these results, meta heuristics are very popular like the well-known Lloyd's algorithm, $k$-means++ \cite{av2007}
 and $k$-means$\|$ \cite{bmvkv2012}. Although $k$-means++ is proved to be $O(\log k)$-approximate, we believe it is far from the limit. The first constant approximation for $k$-means in general dimension is a $(9+\epsilon)$-approximation  \cite{kmn2002} based on local search, which is proved to be almost tight by showing that any local search with a fixed number of swap operations achieves at least a $(9-\epsilon)$-approximation. Ahmadian et al.  \cite{ans2017} improve this factor by presenting a new primal-dual approach and achieve an approximation ratio of 6.357. Recent results  \cite{frs2016, ckm2016} show local search allowing up to $1/\epsilon^{O(1)}$ many swaps in a single neighbor search move yields a PTAS for $k$-means in fixed dimensional Euclidean space.

Capacity constraints form a straightforward variant as well as a natural requirement for almost all combinatorial models, such as the capacitated supply chain, capacitated lot-sizing, capacitated facility location, capacitated vehicle routing, etc. However, these constraints always raise the difficulties of the problem dramatically. For clustering problems like facility location and $k$-median, the capacitated version breaks the nice structure of an optimal solution that every data point is assigned to its nearest open center. Moreover, the widely used integer-linear programs for these problem have unbounded integral gap when the capacity constraints are added. The state-of-the-art approximation algorithms also suffer from this limitation. Researchers have to change techniques instead of insisting on linear program based ones. Note that the linear program based techniques perform very well in the uncapacitated facility location and $k$-median.

As stressed before, approximation algorithms are seldom proposed for capacitated $k$-means due to its hardness. Fortunately, this problem has strong connections with other capacitated clusterings like capacitated facility location and capacitated $k$-median. For capacitated facility location, the state-of-the-art approximation algorithm achieves a $(5.83+\epsilon)$-approximation ratio based on local search \cite{zcy2005}, while for uncapacitated facility location the corresponding result is $1.488$ \cite{l2013} which is very close to the inapproximability lower bound of $1.463$ under the assumption of P $\neq$ NP. For the capacitated $k$-median, most research in the literature concerns pseudo-approximation that violates either the cardinality constraint or capacity constraints. Based on an exploring of extended relaxations for the capacitated $k$-median, Li S. \cite{l2017} presents an $e^{O(1/\epsilon^2)}$-approximation that violates the cardinality constraint by a factor of $1+\epsilon$. Byrka et al. \cite{br2016} present a constant approximation with a $1+\epsilon$ violation of capacity constraint. They are both the state-of-the-art bi-factor approximation results for uniform capacitated $k$-median. Recently, a combination of approximation algorithms and FPT algorithms has been developed for those problems, for which we fail to give polynomial-time constant approximations. Particularly, Adamczyk et al. \cite{abm2018} propose a constant FPT approximation for both uniform and nonuniform capacitated $k$-median, which inspires our work.

\noindent \textbf{Our contribution}
We present a constant FPT approximation algorithm without breaking any cardinality constraint or capacity constraint for any general HCKM instance. The running time of the proposed algorithm is polynomial w.r.t. the inputs except for $k$. Also, our algorithm has the following advantages and potential: (1) It is embedded into two main subroutines which allow to be replaced by others so that one can find a better balance between performance ratio and running time under this framework; (2) The algorithm as well as the analysis are very likely to extend to solving a large family of capacitated clustering with objective to minimize the sum of arbitrary powers of Euclidean distances; (3) The technique may inspire further work of nonuniform capacitated $k$-means, for which, to the best of our knowledge, even no appropriate mathematical model was known so far.

The remainder of the paper is organized in a natural way that first proposes the algorithm, and continues with the running time and performance analysis.

\section{FPT Algorithm}
The KM can be formally described as: We are given a data set $X=\{x_1, x_2, \dots, x_n \}$, where $x_i$ $(i\in\{1, 2, \dots, n\})$ are $d$-dimensional real vectors. The object is to partition the data set into $k(\le n)$ disjoint subsets so as to minimize the total within-cluster sum of squared distances (or variances). It is well known that for a fixed finite cluster $A\subseteq  \mathbb{R}^d$, the variance minimizer must be the centroid of $A$, which we denote by ${\rm ctr} (A):= \frac{1}{|A|}\sum_{x \in A} x.$ Thus the goal of the $k$-means can be formally stated as to find a partition $\{X_1, X_2, \dots, X_k\}$ of $X$ so as to minimize the following objective:
$$\sum\limits^k_{i=1}\sum\limits_{x\in X_i} \left\|x- {\rm ctr} (X_i)\right\|^2.$$
In the HCKM, we force the size of each cluster to be smaller than a given input $u$, i.e. $|X_i|\le u$ for $i=1, 2, \dots, k$. For the sake of convenience, we assume w.l.o.g. that $u$ is a positive integer.

Here is an easy observation for an arbitrary optimal solution to KM. Suppose we are given the optimal center set of cardinality $k$, then the rest of the problem would be easy. Since the object is to minimize the total squared distances from each data point to the center set, it must be the case that every data point is assigned to its nearest center in the given set, which is also known as Voronoi partition. However, this partition may violate the capacity constraints and thus be infeasible to HCKM. For this consideration, we claim the following.

\begin{lemma}\label{LP}
When given the center set, the HCKM is polynomial-time tractable.
\end{lemma}
\begin{proof}
Suppose we have an $n$-point data set $X=\{x_1,x_2,\cdots,x_n\}\subseteq \mathbb{R}^d$ and an integer $u$ as input. And we are given $k$-point center set $C=\{c_1,c_2,\cdots,c_k\}\subseteq \mathbb{R}^d$. Then the optimal partition of $X$ according to $C$ satisfying the cluster upper bound constraints $u$ can be computed in polynomial time w.r.t. $n$ and $k$. In fact, the above partition problem can be described by the following integer-linear program.
\bea \label{TP}
\min\limits_{p_{ij}}  & &   \sum\limits_{i=1}^k\sum\limits_{j=1}^np_{ij}\cdot dist(c_i,x_j)  \\
{\rm s. t.} & &  \sum\limits_{i=1}^kp_{ij}=1, ~~~~\forall~ j=1,\cdots,n  \\
& & \sum\limits_{j=1}^np_{ij}\le u, ~~~~\forall~ i=1,\cdots,k  \\
& & p_{ij}\in\{0,1\},~~~~\forall~ i=1,\cdots,k, ~~j=1,\cdots,n
\eea

The decision variable $p_{ij}$ (out of an arbitrary solution $p$) equals 1 representing that $x_j$ is assigned to $c_i$ in that solution, and 0 otherwise. In HCKM, once we have the location information of any two points $x,y\in \mathbb{R}^d$, the distance between them can be computed immediately by $dist(x,y)=\|x-y\|^2$. Thus the objective is essentially linear w.r.t. $p_{ij}$. The constraints of the above program are quite straightforward.

It seems to be intractable to solve the program exactly, but one can find that the constraints as well as the object are linear except the Boolean constraints. From the theory of integer-linear program we know that at least one of the optimal solutions to the natural relaxation of the above program must be integral (Detail see \cite{sa2003} as a reference). Note that the relaxation of program (1-4) differs only by relaxing the constraints $p_{ij}\in\{0,1\}$ to $p_{ij}\in [0,1]$. The integrality property of the relaxed program holds whenever the coefficient matrix of the initial integer-linear program is totally unimodular and the right-hand side terms are integers, which both are satisfied in program (1-4). That is to say, we only need to compute an optimal solution of a linear program in order to obtain the optimal assignment, implying the lemma.
\end{proof}


For any HCKM instance $X$, $u$ and $k$, let $\Pi(C)$ illustrate the optimal partition w.r.t. an arbitrary center set $C$. According to the above lemma, $\Pi(C)$ can be computed in polynomial time. For uncapacitated $k$-means instances, we use $\Pi_{vor}(C)$ to denote the Voronoi partition w.r.t. $C$. Moreover, let $\pi$ and $\pi_{vor}$ be the corresponding mapping from $X$ to $C$ in  partition $\Pi$ and $\Pi_{vor}$ respectively. Thus $\pi(x)=c$ and $\pi^{-1}(c)=x$ if and only if  $x$ is assigned to $c$ in $\Pi$ and $\pi_{vor}$ similarly. Therefore, once we have the center set $C$ (not necessary optimal) for HCKM, we can compute the objective value w.r.t. $C$ by
 $$cost(C):=\sum\limits_{c\in C}\sum\limits_{x\in \pi^{-1}(c)} \left\|x- c\right\|^2.$$

The rest of the HCKM, also the critical part, it to find the optimal center set $C$. Let $n:=|X|$, there are $\binom{2^n}{k}$ many candidates of the centroid set for KM  and this amount do not reduce too much for that of HCKM. We know from practical experience that the number of clusters $k$ is literally much smaller than the total number of data $n$. Thus it is interesting if we are able to reduce the number of candidates to $\binom{O(k)}{k}$. This main idea reminds us of the bi-criteria approximation algorithm for KM that partitions the data set into $O(k)$ clusters, which will be considered as a subroutine in our algorithm. We call the outputs of the KM subroutine the representing set, denoted by $S$. At the end of this subroutine, we actually obtain $O(k)$ regions in $\mathbb{R}^d$ according to the Voronoi partition of the representing set, see Fig.\ref{structure}. By introducing a novel metric $\mathbb{H}$, we are able to solve the optimal solution to HCKM efficiently. And by carefully analysis, we prove that we only lose a constant factor of the optimum w.r.t. the original metric.

\begin{figure}[htb]
  \begin{center}
  \includegraphics[width=3.5in]{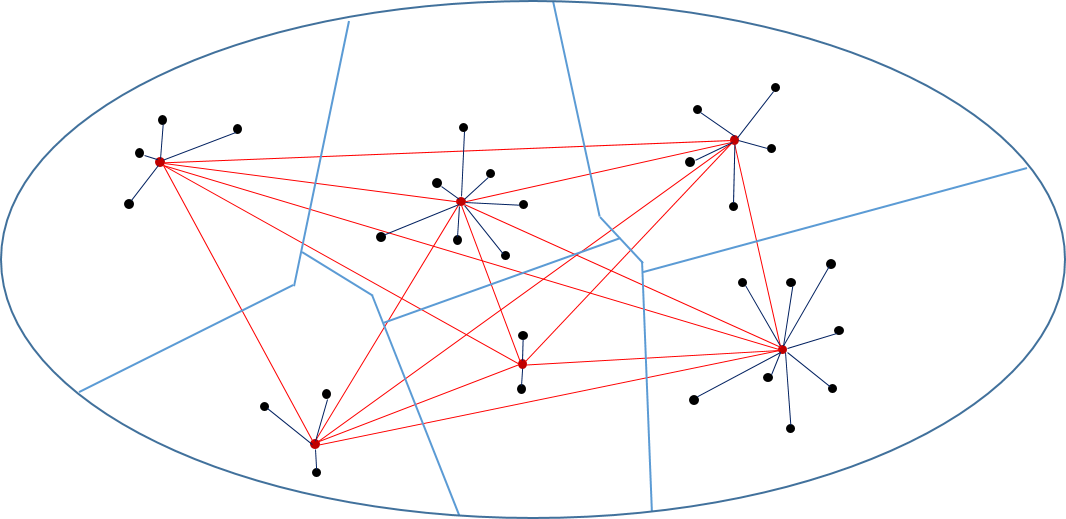}\\
  \end{center}
  \caption{The structure at the end of the KM subroutine}\label{structure}
\end{figure}

Let metric $\mathbb{D}$ be the $d$-dimensional Euclidean space equipped with distance $d(x,y)=\|x-y\|^2$ for any two point $x,y\in \mathbb{R}^d$, where $\|\cdot\|$ denotes the standard Euclidean distance or 2-norm distance. Under this definition, the object for HCKM is simply the sum of distance between each data point with its centroid under the distance in metric $\mathbb{D}$. However, it is hard to compute the optimal centroid set under this metric. But for the following metric (see Fig. \ref{pq} as an example), it is not the case.

\begin{definition}
 Define metric $\mathbb{H}$ be $\mathbb{R}^d$ equipped with the following distance $$h(x,y):=d(x, \pi_{vor}(x))+d(\pi_{vor}(x), \pi_{vor}(y))+d(\pi_{vor}(y),y), ~~ \forall x,y\in \mathbb{R}^d,$$ where $\pi_{vor}(x)$ and $\pi_{vor}(y)$ are the corresponding nearest points (w.r.t. the distance in metric $\mathbb{D}$) in $S$.
\end{definition}

\begin{figure}[htb]
  \begin{center}
  \includegraphics[width=2.5in]{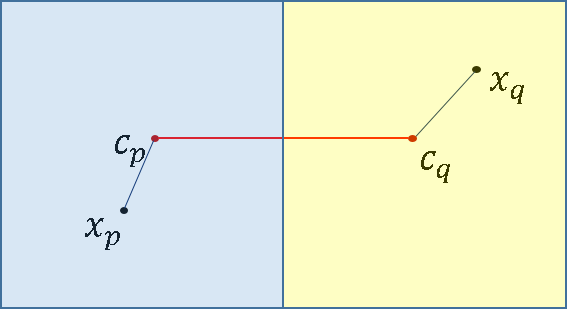}\\
  \caption{The distance between $x_p$ and $x_q$ in metric $\mathbb{H}$}\label{pq}
   \end{center}
\end{figure}

Under this novel metric, think about the case that an arbitrary data point $x$ that is assigned to a certain region among the $O(k)$ regions. It is must be the case for any optimal solution that it is assigned to the representing point in that region. Suppose data point $a$ is assigned to $b$ instead of the representing point in which region $b$ is located, then one can reduce the cost by move $b$ to $\pi_{vor}(b)$ because $h(a,b)=h(a,\pi_{vor}(b))+d(b,\pi_{vor}(b))$. Next we will use notation $cost_{\mathbb{H}}(\cdot)$ representing the objective value in metric $\mathbb{H}$ and $cost_{\mathbb{D}}(\cdot)$ in metric $\mathbb{D}$ for HCKM. And let $opt(\mathbb{H})$ together with $opt(\mathbb{D})$ be  the corresponding optimal solution for HCKM in metric $\mathbb{H}$ and metric $\mathbb{D}$. Recall $S$ is the output solution of the (uncapacitated) KM subroutine in metric $\mathbb{D}$ and $cost_{\mathbb{D}}(S)$ is its cost in metric $\mathbb{D}$. Note there is no difference between the objective value of KM and that of HCKM. The optimal solution for KM in metric $\mathbb{D}$ is denoted by $opt_{KM}(\mathbb{D})$. Again, all these notations are based on the same fixed instance.

Therefore, the optimal centroid set under metric $\mathbb{H}$ must be located in the representing set. Note once we have the optimal centroid set, we can in polynomial time compute the optimal partition by solving a linear program. The rest of the problem is easy---select the $k$ optimal centroid out of $O(k)$ candidates. This is another subroutine in our algorithm, which we take a simple but efficient one, the brute force search, i.e. to enumerate all possible distributions.

Since we do not require the cardinality of $k$ but care about the constant approximation here, we employ Hsu's Algorithm \cite{ht2016} as an example that opens at most $O(k\log \epsilon^{-1})$ centers achieving $(1+\epsilon)$ approximation. Note in low-dimensional space, Bandyapadhyay and Varadarajan's algorithm \cite{bv2015} that opens at most $(1+\epsilon)k$ centers with $(1+\epsilon)$ approximation is also a good choice for our FPT algorithm. To deal with the cardinality constraint, we make a guess over the distribution of the $k$ centers among the $O(k)$ clusters which is obtained in the KM subroutine. The optimal centroid set must be the one with minimum partition cost. For practical consideration, we add an update step that replace the current centroid set by the real centroid set of current solutions. This step definitely reduce the cost but does not change the partition.

By the above, we actually compute an optimal solution to HCKM in metric $\mathbb{H}$. In next section we prove that by doing this, we only lose a constant factor of the minimum cost. For ease of the analysis, we present the pseudocode of our FPT algorithm by Algorithm \ref{alg1}.
\begin{algorithm}[H]
\caption{Pseudocode of the FPT approximation algorithm for HCKM}
\label{alg1}
\begin{algorithmic}[1]
\Require
$n$-point data set $X\subseteq \mathbb{R}^d$, integers $k$ and $u$, small $\epsilon >0$
\Ensure
Partition of $X$
\begin{description}
\item[\textbf{~~1:}] \textbf{if} $k> n$ or $u< n/k$
  \item[\textbf{~~2:}]~~~ \textbf{return} "Infeasible instance";
  \item[\textbf{~~3:}] \textbf{Otherwise} \\
  ~~~~~~~Call KM subroutine (with parameter $\epsilon^\prime=\epsilon/36$) on $X$ obtaining representing set $S=\{s_1,\cdots,s_{|S|}\}$ \\~~~~~where $|S|=O(k\log \epsilon^{-1})$;
  \item[\textbf{~~4:}] ~~~~\textbf{for} each $|S|$-dimensional nonnegative integer vector $p$ such that $\|p\|_1=k$
  \item[\textbf{~~5:}] ~~~~~~~~~\textbf{for} each $i=1:|S|$
  \item[\textbf{~~6:}] ~~~~~~~~~~~~~~~ Add $p_i$ times $s_i$ to $C_i^p$;
  \item[\textbf{~~7:}] ~~~~~~~~~\textbf{end for}
  \item[\textbf{~~8:}] ~~~~~~~~~$C^p$ $\leftarrow$ $\bigcup_{i=1}^{|S|}C_i^p$;
  \item[\textbf{~~9:}] ~~~~\textbf{end for}
  \item[\textbf{10:}] ~~~~$C\leftarrow \arg\min\limits_p cost_{\mathbb{H}}(C^p)$;
  \item[\textbf{11:}] ~~~~\textbf{return} $\Pi(C)$;
  \item[\textbf{12:}]~~~~update $C$;
  \item[\textbf{13:}]\textbf{end if}

\end{description}
\end{algorithmic}
\end{algorithm}

\section{Analysis}
We evaluate the proposed algorithm from both running-time and performance ratio perspectives. For ease of the notation, we define the squared Euclidean distance between location $a\in \mathbb{R}^d$ and set $B \subseteq \mathbb{R}^d$ by $d(a, B)=\min_{b\in B}d(a,b)$.
And remember $d(a,b)=\|a-b\|^2$ for any $a, b\in \mathbb{R}^d$.

\subsection{Running-time Analysis}
By the following lemma we show the running-time of our algorithm is only exponential w.r.t. $k$ and polynomial w.r.t. all other inputs.
\begin{lemma}
\label{time}
Algorithm \ref{alg1} terminates in time $2^{O(k\log k)}n^{O(1)}$.
\end{lemma}
\begin{proof}
We begin with the subroutine that employs the Hsu's Algorithm for KM which terminates in time $O(dk\log\epsilon^{-1}n^{1+\lceil\epsilon^{-1}\rceil})$. It is polynomial w.r.t. inputs $d,k$ and $n$ for any fixed constant $\epsilon$. Thus the running-time of Algorithm \ref{alg1} is dominated by the loop from step 4 to 10.

 In step 3 we divide the whole $\mathbb{R}^d$ space into a Voronoi partition by $S$ output by the KM subroutine. What we do next in the "for" loop is to numerate all possible configurations of the number distribution for the centroid set over all regions. Ignoring the effect of $\epsilon$, it scans at most $(O(k))^k$ probabilities. And the inner loop in step 5 takes $O(k)$ numerations. In each numeration we pick a known number of copies of centers in the correlative region which exactly is what step 6 all about. After the inner loop, we actually do a more time-consuming work of solving a large number (probabilities of $p$) of linear programs when computing $\arg\min_p cost(C^p)$. The linear program is the relaxation of integer-linear program (1-4) in which there are $O(nk)$ variables as well as constraints. Given $k\le n$, we bound the solving time by $n^{O(1)}$. Therefore, this dominating step takes $(O(k))^kn^{O(1)}$ in total, implying the theorem.
\end{proof}

\subsection{Performance Analysis}
Let us take a glance at the Voronoi partition at the end of step 3. We get $O(k)$ representing points as well as regions from the KM subroutine. Among them there are $k$ points supposed to be open in later steps.  Some dense regions may open more than once while some may open none.
But the proposed algorithm only open $k$ out of $O(k)$ regions, thus in average there must be some regions where no representing point will be open. Therefore the data points in that region must be assigned to the open representing point in other region and form a new cluster. We will focus on these data points and try to bound their assignment cost. Next we introduce an inequality that will be used frequently in our analysis. We all know that the triangle inequality holds in Euclidean space. We prove a similar property holding in squared Euclidean space which we call the extended triangle inequality.

\begin{lemma}[Extended triangle inequality]
For any $i,j,k\in\mathbb{R}^d$, we have $$d(i,j)\le 2(d(i,k)+d(j,k)).$$
\end{lemma}
\begin{proof}
Based on basic algebra facts, we have $$\frac{(\sqrt{d(i,j)})^2}{2}\le \frac{(\sqrt{d(i,k)}+\sqrt{d(j,k)})^2}{2}\le d(i,k)+d(j,k),$$ where the first inequality holds because $\sqrt{d(x,y)}$ is actually the Euclidean distance between $x$ and $y$ and thus satisfying the triangle inequality. Implies the lemma.
\end{proof}

Moreover, the extended triangle inequality can be extended to include more complicated cases.

\begin{corollary}\label{corollary}
If $i, j\in \mathbb{R}^d$ are inserted in two points l and k in metric $\mathbb{D}$, it must be the case that
$$d(i,j)\le 3(d(i,l)+d(l,k)+d(k,j)).$$
\end{corollary}
\begin{proof}
Similarly, $$\frac{(\sqrt{d(i,j)})^2}{3}\le \frac{(\sqrt{d(i,l)}+\sqrt{d(l,k)}+\sqrt{d(k,j)})^2}{3}\le d(i,l)+d(l,k)+d(k,j),$$ where the first inequality comes from the triangle inequality and second from basic algebra facts.
\end{proof}

Now think about the relation between the optimum value for HCKM in metric $\mathbb{H}$ and that in metric $\mathbb{D}$.
By the following lemma, we show $cost_{\mathbb{H}}(opt(\mathbb{D}))$ can be bounded in terms of $cost_{\mathbb{D}}(opt(\mathbb{D}))$ and $cost_{\mathbb{D}}(S)$ from both sides.

\begin{lemma}\label{KM}
$1/3cost_{\mathbb{D}}(opt(\mathbb{D}))\le cost_{\mathbb{H}}(opt(\mathbb{D})) \le 11cost_{\mathbb{D}}(opt(\mathbb{D}))+12cost_{\mathbb{D}}(S)$
\end{lemma}
\begin{proof}
Suppose $x_p$ is assigned to $x_q$ in $opt(\mathbb{D})$. Let $c_p$ ($c_q$) be the nearest center from $x_p$ ($x_q$) to $S$ which is obtained through the KM subroutine, i.e., $c_p=\pi_{vor}(x_p)$, $c_q=\pi_{vor}(x_q)$. (See Fig. \ref{pq})
The lower bound is straightforward because from Corollary \ref{corollary} we have,
$$d(x_p,x_q)\le 3(d(x_p, c_p)+d(c_p, c_q)+d(x_q,c_q))=3h(x_p,x_q).$$
Because of the $x_q=\pi(x_p)$ assumption, summation over all $x_p\in X$ obtaining
$$1/3cost_{\mathbb{D}}(opt(\mathbb{D}))\le cost_{\mathbb{H}}(opt(\mathbb{D})).$$
For the upper bound, an observation of Voronoi partition gives $d(x_q,c_q)\le d(x_q,c_p)$ and thus
$$d(x_q,c_q)\le d(x_q,c_p)\le 2(d(x_p,c_p)+d(x_p,x_q)),$$
where the second inequality follows from the extended triangle inequality. On the other hand, we can bound $d(c_p,c_q)$ using Corollary \ref{corollary} and thus,
$$d(c_p,c_q)\le 3(d(x_p,c_p)+d(x_p,x_q)+d(x_q,c_q))\le 9(d(x_p,c_p)+d(x_p,x_q)).$$
So in total, $$h(x_p,x_q)\le 11d(x_p,x_q)+12d(x_p,c_p).$$
Again, summation over all $x_p\in X$ and remember $c_p=\pi_{vor}(x_p)$, $c_q=\pi_{vor}(\pi(x_p))$,
$$cost_{\mathbb{H}}(opt(\mathbb{D})) \le 11cost_{\mathbb{D}}(opt(\mathbb{D}))+12cost_{\mathbb{D}}(S),$$
completing the whole proof.
\end{proof}

In fact, for any feasible solution $\sigma$ to HCKM, we always have  $cost_{\mathbb{D}}(\sigma)\le 3cost_{\mathbb{H}}(\sigma)$. By the above lemma we only prove the special case where $\sigma$ is an optimal one. Blending this observation with Lemma \ref{LP}, we state the following.

\begin{lemma}\label{HCKM}
Algorithm \ref{alg1} outputs a feasible solution to HCKM with cost $$cost_{\mathbb{D}}(C)\le 3cost_{\mathbb{H}}(opt(\mathbb{H})).$$
\end{lemma}
\begin{proof}
The feasibility is straightforward. For the cardinality constraint, from step 4 we have $\|p\|_1=\sum\limits_{i=1}^{|S|}p_i=k$. Thus all $C^p$ w.r.t. vector $p$ have cardinality $$|C^p|=|\bigcup_{i=1}^{|S_i|}C_i^p|=\sum\limits_{i=1}^{|S|}|C_i^p|=\sum\limits_{i=1}^{|S|}p_i=k.$$
$C$ is one of $C^p$ thus satisfying the cardinality constraint. The capacity constraints hold naturally because $\Pi(C)$ is obtained by solving program (1-4).

We already claim that $cost_{\mathbb{D}}(C)\le 3cost_{\mathbb{H}}(C)$ which is derived directly from Corollary \ref{corollary}. To complete the rest of the proof, we only need to prove $cost_{\mathbb{H}}(C)=cost_{\mathbb{H}}(opt(\mathbb{H}))$. In fact, $\Pi(C)$ is an exact optimal solution for HCKM in metric $\mathbb{H}$. That is, $\Pi(C)$ is a such a partition of $X$ ($\pi$ is the correlative mapping) satisfying both cardinality and capacity constraints and at the same time minimize the following,
$$ \sum\limits_{i=1}^k\sum\limits_{x\in\pi^{-1}(c_i)}h(x,c_i)=\sum\limits_{i=1}^k\sum\limits_{x\in\pi^{-1}(c_i)}\left(d(x,\pi_{vor}(x))+d(c_i,\pi_{vor}(x))+d(c_i,\pi_{vor}(c_i))\right).$$

First, we prove that the optimal solution for HCKM in metric $\mathbb{H}$ must be a subset of $S$ (probably a multi-set). By contradiction, suppose there exist an optimal solution $C^\prime$ with a center $c^\prime$ not locating in $S$. And w.l.o.g. assume the optimal assignment mapping w.r.t. $C^\prime$ is $\pi^{\prime}$. Considering the cluster $(\pi^{\prime})^{-1}(c^\prime)$, we can reduce the cost of this cluster by moving $c^\prime$ to $\pi_{vor}(c^\prime)$. Because from the definition of metric $\mathbb{H}$ we have,
\begin{eqnarray}
\begin{aligned}
\sum\limits_{x\in(\pi^{\prime})^{-1}(c^\prime)}h(x,\pi_{vor}(c^\prime))=&\sum\limits_{x\in(\pi^{\prime})^{-1}(c^\prime)}h(x,c^\prime)-|(\pi^{\prime})^{-1}(c^\prime)|d(c^\prime,\pi_{vor}(c^\prime))\\
\le & \sum\limits_{x\in(\pi^{\prime})^{-1}(c^\prime)}h(x,c^\prime) \nonumber
\end{aligned}
\end{eqnarray}

Known the optimal solution must be a subset of $S$, we enumerate all possible subsets for which  we compute the optimal assignment by program (1-4) with $dist(c_i,x_j)=h(c_i,x_j)$. Thus $\Pi(C)$ is one of the optimal solutions in metric $\mathbb{H}$ and naturally $cost_{\mathbb{H}}(C)=cost_{\mathbb{H}}(opt({\mathbb{H}}))$. Combining with $cost_{\mathbb{D}}(C)\le 3cost_{\mathbb{H}}(C)$, implies the lemma.
\end{proof}

For the sake of portability and completeness, we propose a more general framework showing that once we have a $\lambda_1$-approximation KM subroutine as well as a $\lambda_2$-approximation subroutine for HCKM in metric $\mathbb{H}$, we can combine them to obtain a $3\lambda_2(11+12\lambda_1)$-approximation algorithm for HCKM in metric $\mathbb{D}$.

\begin{lemma}\label{frame}
Suppose we have a $\lambda_1$-approximation KM subroutine for the uncapacitated $k$-means problem that outputs a solution $S$ with $O(k)$ opened centers, as well as a $\lambda_2$-approximation subroutine for HCKM in metric $\mathbb{H}$ that outputs a solution $C$ with $k$ opened centers. Then, we can construct an $3\lambda_2(11+12\lambda_1)$-approximation algorithm (not necessary in polynomial time) for any general HCKM instances.
\end{lemma}
\begin{proof}
Suppose we have two subroutines that output $S$ for KM and $C$ for HCKM satisfying
\begin{center}
$cost_{\mathbb{D}}(S)\le \lambda_1 cost_{\mathbb{D}}(opt_{KM}(\mathbb{D}))$~ and ~$cost_{\mathbb{H}}(C)\le \lambda_2 cost_{\mathbb{H}}(opt(\mathbb{H}))$.
\end{center}
Now we need to modify Lemma \ref{KM} and \ref{HCKM} according to the above conditions. First, observe that any feasible solution to KM is also feasible to HCKM and thus
$cost_{\mathbb{D}}(opt_{KM}(\mathbb{D}))\le cost_{\mathbb{D}}(opt(\mathbb{D}))$. Combing this with $\lambda_1$-approximation condition we know,
$cost_{\mathbb{D}}(S)\le \lambda_1cost_{\mathbb{D}}(opt(\mathbb{D})).$
Substitute into Lemma \ref{KM} we obtain,
$$cost_{\mathbb{H}}(opt(\mathbb{D})) \le 11cost_{\mathbb{D}}(opt(\mathbb{D}))+12\lambda_1cost_{\mathbb{D}}(opt(\mathbb{D})).$$
Considering Lemma \ref{HCKM}, we replace the inequality $cost_{\mathbb{H}}(C)=cost_{\mathbb{H}}(opt(\mathbb{H}))$ with $cost_{\mathbb{H}}(C)\le \lambda_2 cost_{\mathbb{H}}(opt(\mathbb{H}))$ in the proof and get
$$cost_{\mathbb{D}}(C)\le 3\lambda_2cost_{\mathbb{H}}(opt(\mathbb{H})).$$
Blending modified Lemma \ref{KM} and \ref{HCKM} together,
$$cost_{\mathbb{D}}(C)\le 3\lambda_2cost_{\mathbb{H}}(opt(\mathbb{H})) \le3\lambda_2cost_{\mathbb{H}}(opt(\mathbb{D})) \le 3\lambda_2(11+12\lambda_1)cost_{\mathbb{D}}(opt(\mathbb{D})),$$
completing the proof.
\end{proof}

Note Algorithm \ref{alg1} employs a $(1+\epsilon^\prime)$-approximation KM subroutine and an exact subroutine for HCKM in metric $\mathbb{H}$, where $\epsilon^\prime=\epsilon/36$. Substitute $\lambda_1=1+\epsilon/36, \lambda_2=1$ into Lemma \ref{frame} implying a $(69+\epsilon)$-approximation for general HCKM. We obtain the following analysis result for our algorithm.

\begin{theorem}
Algorithm \ref{alg1} outputs a feasible solution to HCKM with approximation ratio $69+\epsilon$ in time $2^{O(k\log k)}n^{O(1)}$ .
\end{theorem}

\section{Conclusion and future work}
The proposed algorithm is a theoretical work that does not violate any cardinality and capacity constraints, and at the same time achieves a constant approximation ratio. The algorithm allows one to balance the performance ratio and running-time by embedding different subroutines. Concerning the performance analysis, we employ the Hsu's KM subroutine and an exact subroutine for computing the optimal partition in metric $\mathbb{H}$ at the cost of time consumption. One can find a better balance in this framework. Besides, our analysis is a rough one that only concerns to achieve a constant performance ratio. One can easily improve the ratio as well as the time analysis as we ignore many low order terms. A faster approximation algorithm other than an FPT one is also an interesting direction.




\bibliography{lipics-v2018-sample-article}

\end{document}